\documentclass[letterpaper]{article} 
\usepackage{aaai23}  
\usepackage{times}  
\usepackage{helvet}  
\usepackage{courier}  
\usepackage[hyphens]{url}  
\usepackage{graphicx} 
\usepackage{natbib}  
\usepackage{caption} 
\usepackage{booktabs}
\usepackage{setspace}
\usepackage{amsthm}
\usepackage{amssymb}
\usepackage{amsmath}
\usepackage{xcolor}
\usepackage{float}
\usepackage{accents}
\usepackage{algorithm}
\usepackage[noend]{algorithmic}
\usepackage{subcaption}
\usepackage{thmtools}
\usepackage{thm-restate}
\usepackage{nccmath}
\usepackage{enumitem}
\usepackage[symbol]{footmisc}
\usepackage{adjustbox}
\usepackage{multirow}

\usepackage{newfloat}
\usepackage{listings}
\DeclareCaptionStyle{ruled}{labelfont=normalfont,labelsep=colon,strut=off} 
\floatstyle{ruled}
\newfloat{listing}{tb}{lst}{}
\floatname{listing}{Listing}
\newcommand{\ls}{\mathnormal{U}}
\newcommand{\rs}{\mathnormal{V}}

\newcommand{\lv}{\mathnormal{u}}
\newcommand{\rv}{\mathnormal{v}}

\newcommand{\wto}{\mathnormal{w^{O}_e}}
\newcommand{\wtl}{\mathnormal{w^{\ls}_e}}
\newcommand{\wtr}{\mathnormal{w^{\rs}_e}} 

\newcommand{\pl}{\mathnormal{\Delta_{\lv}}}
\newcommand{\pr}{\mathnormal{\Delta_{\rv}}} 

\newcommand{\pv}{\mathnormal{p_{v}}}
\newcommand{\pvt}{\mathnormal{p_{v,t}}}
\newcommand{\nv}{\mathnormal{n_{v}}}

\newcommand{\xet}{\mathnormal{x_{e,t}}}

\newcommand{\xets}{\mathnormal{x^*_{e,t}}}
\newcommand{\yets}{\mathnormal{y^*_{e,t}}}
\newcommand{\zets}{\mathnormal{z^*_{e,t}}}

\newcommand{\ret}{\mathnormal{\rho_{e,t}}}

\newcommand{\grps}{\mathcal{G}}

\newcommand{\ah}{\mathnormal{\hat{\alpha}}}
\newcommand{\bh}{\mathnormal{\hat{\beta}}}
\newcommand{\gh}{\mathnormal{\hat{\gamma}}}

\newcommand{\ag}{\mathnormal{\alpha_{G}}}
\newcommand{\ai}{\mathnormal{\alpha_{I}}}

\newcommand{\floor}[1]{\left\lfloor #1 \right\rfloor}
\newcommand{\ceil}[1]{\left\lceil #1 \right\rceil}

\newtheorem{theorem}{Theorem}[section]

\newtheorem{corollary}{Corollary}[section]

\DeclareMathOperator{\OPT}{\mathrm{OPT}}
\DeclareMathOperator{\ALG}{\mathrm{ALG}}

\DeclareMathOperator{\tsf}{\mathrm{TSGF}_{\textbf{KIID}}}
\DeclareMathOperator{\tsfkad}{\mathrm{TSGF}_{\textbf{KAD}}}

\DeclareMathOperator{\E}{\mathbb{E}}
\DeclareMathOperator{\EI}{\mathbb{E}_{\mathcal{I}}}

\newcommand{\PPDR}{\textbf{PPDR}}

\newcommand{\hide}[1]{}

\title{Rawlsian Fairness in Online Bipartite Matching: Two-Sided, Group, and Individual}
\author {
    Seyed Esmaeili \textsuperscript{\rm 1},
    Sharmila Duppala \textsuperscript{\rm 1},
    Davidson Cheng \textsuperscript{\rm 2},
    Vedant Nanda \textsuperscript{\rm 1},
    Aravind Srinivasan \textsuperscript{\rm 1},
    John P. Dickerson \textsuperscript{\rm 1}
}
\affiliations {
    \textsuperscript{\rm 1} University of Maryland, College Park\\
    \textsuperscript{\rm 2} Colorado College\\
    \{esmaeili,sduppala,vedant,srin,john\}@cs.umd.edu,  d\_cheng@coloradocollege.edu
}

\begin{document}

\maketitle

\begin{abstract}
Online bipartite-matching platforms are ubiquitous and find applications in important areas such as crowdsourcing and ridesharing. In the most general form, the platform consists of three entities: two sides to be matched and a platform operator that decides the matching. The design of algorithms for such platforms has traditionally focused on the operator’s (expected) profit. Since fairness has become an important consideration that was ignored in the existing algorithms a collection of online matching algorithms have been developed that give a fair treatment guarantee for one side of the market at the expense of a drop in the operator’s profit. In this paper, we generalize the existing work to offer fair treatment guarantees to \emph{both} sides of the market \emph{simultaneously}, at a calculated worst case drop to operator profit. We consider group and individual Rawlsian fairness criteria. Moreover, our algorithms have theoretical guarantees and have adjustable parameters that can be tuned as desired to balance the trade-off between the utilities of the three sides. We also derive hardness results that give clear upper bounds over the performance of any algorithm. 
\end{abstract}

\section{Introduction}\label{sec:intro}
Online bipartite matching has been used to model many important applications such as crowdsourcing \cite{Ho12:Online,tong2016online,dickerson2019balancing}, rideshare~\cite{Lowalekar18:Online,dickerson2021allocation,Ma21:Fairness}, and online ad allocation~\cite{Goel08:Online,mehta2013online}. In the most general version of the problem, there are three interacting entities: two sides of the market to be matched and a platform operator which assigns the matches. For example, in rideshare, riders on one side of the market submit requests, drivers on the other side of the market can take requests, and a platform operator such as Uber or Lyft matches the riders' requests to one or more available drivers. In the case of crowdsourcing, organizations offer tasks, workers look for tasks to complete, and a platform operator such as Amazon Mechanical Turk (MTurk) or Upwork matches tasks to workers.   

Online bipartite matching algorithms are often designed to optimize a performance measure---usually, maximizing overall profit for the platform operator or a proxy of that objective. However, fairness considerations were largely ignored.  
This is troubling especially given that recent reports have indicated that different demographic groups may not receive similar treatment. For example, in rideshare platforms once the platform assigns a driver to a rider's request, both the rider and the driver have the option of rejecting the assignment and it has been observed that membership in a demographic group may cause adverse treatment in the form of higher rejection. Indeed, \cite{woman-driver,false-fair,public} report that drivers could reject riders based on attributes such as gender, race, or disability. Conversely, \cite{rosenblat2016discriminating} reports that drivers are likely to receive less favorable ratings if they belong to certain demographic groups. 
%
A similar phenomenon exists in crowdsourcing~\cite{galperin2017geographical}.
Moreover, even in the absence of such evidence of discrimination, as algorithms become more prevalent in making decisions that directly affect the welfare of individuals~\cite{barocas-hardt-narayanan,dwork2012fairness}, it becomes important to guarantee a standard of fairness. Also, while much of our discussion focuses on the for-profit setting for concreteness, similar fairness issues hold in not-for-profit scenarios such as the fair matching of individuals with health-care facilities, e.g., in the time of a pandemic.

In response, a recent line of research has been concerned with the issue of designing fair algorithms for online bipartite matching. \cite{lesmana2019balancing,ma2022group,xutrade} present algorithms which give a minimum utility guarantee for the drivers at a bounded drop to the operator's profit. Conversely, \cite{nanda2019balancing} give guarantees for both the platform operator and the riders instead. Finally,  \cite{suhr2019two} shows empirical methods that achieve fairness for both the riders and drivers simultaneously but lacks theoretical guarantees and ignores the operator's profit.

Nevertheless, the existing work has a major drawback in terms of optimality guarantees. Specifically, the two sides being matched along with the platform operator constitute the three main interacting entities in online matching and despite the significant progress in fair online matching none of the previous work considers all three sides simultaneously. In this paper, we derive algorithms with theoretical guarantees for the platform operator's profit as well as fairness guarantees for the two sides of the market. Unlike the previous work we not only consider the size of the matching but also its quality. Further, we consider two online arrival settings: the \textbf{KIID} and the richer \textbf{KAD} setting (see Section~\ref{sec:model} for definitions). We consider both group and individual notions of Rawlsian fairness and interestingly show a reduction from individual fairness to group fairness in the \textbf{KAD} setting. Moreover, we show upper bounds on the optimality guarantees of any algorithm and derive impossibility results that show a conflict between group and individual notions of fairness. Finally, we empirically test our algorithms on a real-world dataset.

\section{Related Work}\label{sec:rw}
It is worth noting that similar to our work, \cite{patro2020fairrec} and \cite{basu2020framework} have considered two-sided fairness as well, although in the setting of recommendation systems where a different model is applied---and, critically, a separate objective for the operator's profit was not considered.

Fairness in bipartite matching has seen significant interest recently. The fairness definition employed has consistently been the well-known Rawlsian fairness \cite{rawls1958justice} (i.e. max-min fairness) or its generalization Leximin fairness.\footnote{Leximin fairness maximizes the minimum utility like max-min fairness. However, it proceeds to maximize the second worst utility, and so on until the list is exhausted.} We note that the objective to be maximized (other than the fairness objective) represents operator profit in our setting.

The case of offline and unweighted maximum cardinality matching is addressed by~\cite{garcia2020fair}, who give an algorithm with Leximin fairness guarantees for one side of the market (one side of the bipartite graph) and show that this can be achieved without sacrificing the size of the match. 
Motivated by fairness consideration for drivers in ridesharing,~\cite{lesmana2019balancing} considers the problem of offline and weighted matching. Specifically, they show an algorithm with a provable trade-off between the operator's profit and the minimum utility guaranteed to any vertex in one-side of the market.

Recently, \cite{ma2020group} considered fairness for the online part of the graph through a group notion of fairness. In particular, the utility for a group is added across the different types and is minimized for the group worst off, in rough terms their notion translates to maximizing the minimum utility accumulated by a group throughout the matching. Their notion of fairness is very similar to the one we consider here. However, \cite{ma2020group} considers fairness only on one side of the graph and ignores the operator's profit. Further, only the matching size is considered to measure utility, i.e. edges are unweighted. 

A new notion of group fairness in online matching is considered in \cite{sankar2021matchings}. In rough terms, their group fairness criterion amounts to establishing a quota for each group and ensuring that the matching does not exceed that quota. This notion can be seen as ensuring that the system is not dominated by a specific group and is in some sense an opposite to max-min fairness as the utility is upper bounded instead of being lower bounded. Further, the fairness guarantees considered are one-sided as well.

On the empirical side of fair online matching, \cite{Mattei17:Mechanisms} and~\cite{Lee19:WeBuildAI} give application-specific treatments in the context of deceased-donor organ allocation and food bank provisioning, respectively. More related to our work is that of \cite{suhr2019two,zhou2021subgroup} which consider the rideshare problem and provide algorithms to achieve fairness for both sides of the graph simultaneously, however both papers lack theoretical guarantees and in the case of \cite{suhr2019two} the operator's profit is not considered.

\section{Online Model \& Optimization Objectives} \label{sec:model}
Our model follows that of \cite{mehta2013online,feldman2009online,bansal2010lp,alaei2013online} and others. We have a bipartite graph $G=(\ls,\rs,E)$ where $U$ represents the set of static (offline) vertices (workers) and $V$ represents the set of online vertex types (job types) which arrive dynamically in each round. The online matching is done over $T$ rounds. In a given round $t$, a vertex of type $v$ is sampled from $V$ with probability $\pvt$ with $\sum_{v \in V} \pvt = 1, \forall t \in [T]$ the probability $\pvt$ is known beforehand for each type $v$ and each round $t$. This arrival setting is referred to as the known adversarial distribution (\textbf{KAD}) setting \cite{alaei2013online,dickerson2021allocation}. When the distribution is stationary, i.e. $\pvt = \pv, \forall t \in [T]$, we have the arrival setting of the known independent identical distribution (\textbf{KIID}). Accordingly, the expected number of arrivals of type $v$ in $T$ rounds is $\nv=\sum_{t \in [T]} \pvt$, which reduces to $\nv=T \pv$ in the \textbf{KIID} setting. We assume that $\nv \in \mathbb{Z^{+}}$ for \textbf{KIID} \cite{bansal2010lp}. Every vertex $u$ ($v$) has a group membership,\footnote{For a clearer representation we assume each vertex belongs to one group although our algorithms apply to the case where a vertex can belong to multiple groups.} with $\grps$ being the set of all group memberships; for any vertex $u \in U$, we denote its group memberships by $g(u) \in \grps$ (similarly, we have $g(v)$ for $v \in V$). Conversely, for a group $g$, $U(g)$ ($V(g)$) denotes the subset of $U$ ($V$) with group membership $g$. A vertex $u$ ($v$) has a set of incident edges $E_u$ ($E_v$) which connect it to vertices in the opposite side of the graph. In a given round, once a vertex (job) $v$ arrives, an irrevocable decision has to be made on whether to reject $v$ or assign it to a neighbouring vertex $u$ (where $(u,v) \in E_v$) which has not been matched before. Suppose, that $v$ is assigned to $u$, then the assignment is not necessarily successful rather it succeeds with probability $p_e=p_{(u,v)} \in [0,1]$. This models the fact that an assignment could fail for some reason such as the worker refusing the assigned job. Furthermore, each vertex $u$ has patience parameter $\pl \in \mathbb{Z^{+}}$ which indicates the number of failed assignments it can tolerate before leaving the system, i.e. if $u$ receives $\pl$ failed assignments then it is deleted from the graph. Similarly, a vertex $v$ has patience $\pr \in \mathbb{Z^{+}}$, if a vertex $v$ arrives in a given round, then it would tolerate at most $\pr$ many failed assignments in that round before leaving the system.

For a given edge $e=(\lv,\rv) \in E$, we let each entity assign its own utility to that edge. In particular, the platform operator assigns a utility of $\wto$, whereas the offline vertex $\lv$ assigns a utility of $\wtl$, and the online vertex $\rv$ assigns a utility of $\wtr$. This captures entities' heterogeneous wants. For example, in ridesharing, drivers may desire long trips from nearby riders, whereas the platform operator would not be concerned with the driver's proximity to the rider, although this maybe the only consideration the rider has. Similar motivations exist in crowdsourcing as well. We finally note that most of the details of our model such the \textbf{KIID} and \textbf{KAD} arrival settings as well as the vertex patience follow well-established and pratically motivated model choices in online matching, see Appendix (\ref{app:model_details}) for more details.

Letting $\mathcal{M}$ denote the set of successful matchings made in the $T$ rounds, then we consider the following optimization objectives: 
\begin{itemize}[leftmargin=*]
    \item \textbf{Operator's Utility (Profit):} The operator's expected profit is simply the expected sums of the profits across the matched edges, this leads to $\E[\sum_{e \in \mathcal{M}} \wto]$. 
    
    \item \textbf{Rawlsian Group Fairness:} 
        \begin{itemize}[leftmargin=*]
            \item \textbf{Offline Side:} Denote by $\mathcal{M}_{\lv}$ the subset of edges in the matching that are incident on $\lv$. Then our fairness criterion is equal to $$\min\limits_{g \in \grps} \frac{\E[\sum_{u \in U(g)}(\sum_{e \in \mathcal{M}_{\lv}} \wtl )]}{|U(g)|}.$$
            this value equals the minimum average expected utility received by a group in the offline side $U$. 

            \item \textbf{Online Side:} Similarly, we denote by $\mathcal{M}_{\rv}$ the subset of edges in the matching that are incident on vertex $\rv$, and define the fairness criterion to be $$\min\limits_{g \in \grps} \frac{\E[\sum_{v \in V(g)} (\sum_{e \in \mathcal{M}_{\rv}} \wtr)]}{\sum_{v \in V(g)} \nv}.$$ this value equals the minimum average expected utility received throughout the matching by any group in the online side $V$.
        \end{itemize}
        
    \item \textbf{Rawlsian Individual Fairness:} 
        \begin{itemize}[leftmargin=*]
            \item \textbf{Offline Side:} The definition here follows from the group fairness definition for the offline side by simply considering that each vertex $u$ belongs to its own distinct group. Therefore, the objective is $\min\limits_{u \in U}\E[\sum_{e \in \mathcal{M}_{\lv}} \wtl ]$. 

            \item \textbf{Online Side:} Unlike the offline side, the definition does not follow as straightforwardly. Here we cannot obtain a valid definition by simply assigning each vertex type its own group. Rather, we note that a given individual is actually a given arriving vertex at a given round $t \in [T]$, accordingly our fairness criterion is the minimum expected utility an individual receives in a given round, i.e. $\min\limits_{t \in [T]} \E[\sum_{e \in \mathcal{M}_{\rv_t}} \wtr)]$, where $\rv_t$ is the vertex that arrived in round $t$.  
        \end{itemize}
\end{itemize}

\vspace{-0.2cm}
\section{Main Results}
\paragraph{Performance Criterion:} We note that we are in the online setting, therefore our performance criterion is the competitive ratio. Denote by $\mathcal{I}$ the distribution for the instances of matching problems, then $\OPT(\mathcal{I})=\E_{I \sim \mathcal{I}}[\OPT(I)]$ where $\OPT(I)$ is the optimal value of the sampled instance $I$. Similarly, for a given algorithm $\ALG$, we define the value of its objective over the distribution $\mathcal{I}$ by $\ALG(\mathcal{I})=\E_{\mathcal{D}}[\ALG(I)]$ where the expectation $\E_{\mathcal{D}}[.]$ is over the randomness of the instance and the algorithm. The competitive ratio is then defined as $\min_{\mathcal{I}} \frac{\ALG(\mathcal{I})}{\OPT(\mathcal{I})}$.

In our work, we address optimality guarantees for each of the three sides of the matching market by providing algorithms with competitive ratio guarantees for the operator's profit and the fairness objectives of the static and online side of the market simultaneously. Specifically, for the \textbf{KIID} arrival setting we have:  

\begin{restatable}{theorem}{mainthkiid}\label{main_th_kiid}
For the \textbf{KIID} setting, algorithm $\tsf(\alpha,\beta,\gamma)$ achieves a competitive ratio of $(\frac{\alpha}{2e},\frac{\beta}{2e},\frac{\gamma}{2e})$\footnote[2]{Here, $e$ denotes the Euler number, not an edge in the graph.} simultaneously over the operator's profit, the group fairness objective for the offline side, and the group fairness objective for the online side, where $\alpha,\beta,\gamma>0$ and $\alpha+\beta+\gamma \leq 1$. 
\end{restatable}


The following two theorems hold under the condition that $p_e=1,  \forall e\in E$. Specifically for the \textbf{KAD} setting we have: 
\begin{restatable}{theorem}{mainthkad}\label{main_th_kad}
For the \textbf{KAD} setting, algorithm $\tsfkad(\alpha,\beta,\gamma)$ achieves a competitive ratio of $(\frac{\alpha}{2},\frac{\beta}{2},\frac{\gamma}{2})$ simultaneously over the operator's profit, the group fairness objective for the offline side, and the group fairness objective for the online side, where $\alpha,\beta,\gamma>0$ and $\alpha+\beta+\gamma \leq 1$. 
\end{restatable}


Moreover, for the case of individual fairness whether in the \textbf{KIID} or \textbf{KAD} arrival setting we have: 
\begin{theorem} \label{main_th_indiv}
For the \textbf{KIID} or \textbf{KAD} setting, we can achieve a competitive ratio of $(\frac{\alpha}{2},\frac{\beta}{2},\frac{\gamma}{2})$ simultaneously over the operator's profit, the individual fairness objective for the offline side, and the individual fairness objective for the online side, where $\alpha,\beta,\gamma>0$ and $\alpha+\beta+\gamma \leq 1$.  
\end{theorem}

We also give the following hardness results. In particular, for a given arrival (\textbf{KIID} or \textbf{KAD}) setting and fairness criterion (group or individual), the competitive ratios for all sides cannot exceed 1 simultaneously:  
\begin{restatable}{theorem}{Hardnessth}\label{Hardness_th}
For all arrival models, given the three objectives: operator's profit, offline side group (individual) fairness, and online side group (individual) fairness. No algorithm can achieve a competitive ratio of $(\alpha,\beta,\gamma)$ over the three objectives simultaneously such that  $\alpha+\beta+\gamma > 1$. 
\end{restatable}



It is natural to wonder if we can combine individual and group fairness. Though it is possible to extend our algorithms to this setting. The follow theorem shows that they can conflict with one another: 
\begin{restatable}{theorem}{Hardnessthindivgroup}\label{Hardness_th_indiv_group}
Ignoring the operator's profit and focusing either on the offline side alone or the online side alone. With $\ag$ and $\ai$ denoting the group and individual fairness competitive ratios, respectively. No algorithm can achieve competitive ratios $(\ag,\ai)$ over the group and individual fairness objectives of one side simultaneously such that  $\ag+\ai> 1$. 
\end{restatable}


Finally, we carry experiments on real-world datasets in Section \ref{sec:experiments}.
\vspace{-0.3cm}
\section{Algorithms and Theoretical Guarantees}
Our algorithms use linear programming (LP) based techniques \cite{bansal2010lp,nanda2019balancing,xutrade,brubach2016online} where first a \emph{benchmark} LP is set up to upper bound the optimal value of the problem, then an LP solution is sampled from to produce an algorithm with guarantees. Due to space constraints, all proofs and the technical details of Theorems (\ref{Hardness_th} and \ref{Hardness_th_indiv_group}) are in Appendix (\ref{app:proofs}). 
\vspace{-0.1cm}
\subsection{Group Fairness for the \textbf{KIID} Setting:}\label{sec:kiid}
Before we discuss the details of the algorithm, we note that for a given vertex type $v \in V$, the expected arrival rate $n_v$ could be greater than one. However, it is not difficult to modify the instance by ``fragmenting" each type with $n_v>1$ such that in the new instance $n_v=1$ for each type. This can be done with the operator's profit, offline group fairness, and online group fairness having the same values. Therefore, in what remains for the \textbf{KIID} setting $n_v=1, \forall v \in V$ and therefore for any round $t$, each vertex $v$ arrives with probability $\frac{1}{T}$. It also follows that for a given group $g$, $\sum_{v \in V(g)} n_v = \sum_{v \in V(g)} 1=|V(g)|$.      

For each edge $e=(u,v) \in E$ we use $x_e$ to denote the expected number of probes (i.e, assignments from $u$ to type $v$ not necessarily successful) made to edge $e$ in the LP benchmark. We have a total of three LPs each having the same set of constraints of (\ref{constraints}), but differing by the objective. For compactness we do not repeat these constraints and instead write them once. Specifically, LP objective (\ref{sys_obj}) along with the constraints of (\ref{constraints}) give the optimal benchmark value of the operator's profit. Similarly, with the same set of constraints (\ref{constraints}) LP objective (\ref{driver_obj}) and LP objective (\ref{rider_obj}) give the optimal group max-min fair assignment for the offline and online sides, respectively. Note that the expected max-min objectives of (\ref{driver_obj}) and (\ref{rider_obj}), can be written in the form of a linear objective. For example, the max-min objective of (\ref{driver_obj}) can be replaced with an LP with objective $\max{\eta}$ subject to the additional constraints that $\forall g\in \grps$ ,  $ \eta \leq  \frac{\sum_{u \in U(g)} \sum_{e \in E_u} \wtl x_{e}p_{e}}{|U(g)|}$. Having introduced the LPs, we will use LP(\ref{sys_obj}), LP(\ref{driver_obj}), and LP(\ref{rider_obj}) to refer to the platform's profit LP, the offline side group fairness LP, and the online side group fairness LP, respectively. 

{\small
\begin{align}
    & \textstyle\max \sum_{e \in E}{\wto x_{e}} p_{e} \label{sys_obj} \\
    & \textstyle\max \min\limits_{g \in \grps}  \frac{\sum_{u \in U(g)} \sum_{e \in E_u} \wtl x_{e}p_{e}}{|U(g)|} \label{driver_obj} \\
    & \textstyle\max \min\limits_{g \in \grps}  \frac{\sum_{v \in V(g)} \sum_{e \in E_v} \wtr x_{e}p_{e}}{|V(g)|} \label{rider_obj}  
    \end{align}
    \vspace{-0.5cm}
    \begin{subequations}\label{constraints}
    \begin{align}
     &  \text{s.t} \quad  \forall e \in E:  0 \leq x_{e} \leq 1 \label{constraint : assignment } \\
    &  \textstyle\forall u \in U: \sum_{e \in E_u} x_{e}p_{e} \leq 1  \label{constraint: capacity} \\
    &  \textstyle\forall u \in U: \sum_{e \in E_u} x_{e} \leq \Delta_u  \label{constraint: cancellation} \\ 
    &  \textstyle\forall v \in V: \sum_{e \in E_v} x_{e} p_{e} \leq 1  \label{constraint : arrivalrate} \\
    &  \textstyle\forall v \in V:  \sum_{e \in E_v}x_{e}  \leq \Delta_v   \label{constraint : patience} 
    \end{align}
   \end{subequations}
}
\vspace{-0.37cm}

Now we prove that LP(\ref{sys_obj}), LP(\ref{driver_obj}) and LP(\ref{rider_obj}) indeed provide valid upper bounds (benchmarks) for the optimal solution for the operator's profit and expected max-min fairness for the offline and online sides of the matching.

\begin{restatable}{lemma}{validkiid} \label{valid_kiid}
For the \textbf{KIID} setting, the optimal solutions of LP (\ref{sys_obj}), LP (\ref{driver_obj}), and LP (\ref{rider_obj}) are upper bounds on the expected optimal that can be achieved by any algorithm for the operator's profit, the offline side group fairness objective, and the online side group fairness objective, respectively.  
\end{restatable}

Our algorithm makes use of the dependent rounding subroutine \cite{gandhi2006dependent}. We mention the main properties of dependent rounding. In particular, given a fractional vector $\vec{x} = (x_1,x_2, \dots, x_t)$ where each $x_i \in [0,1]$, let $k= \sum_{i \in [t]}x_i$ , dependent rounding rounds $x_i$ (possibly fractional) to $X_i \in \{0,1\} $ for each $i \in [t]$ such that the resulting vector $\vec{X}= (X_1,X_2, X_3, \dots ,X_t)$ has the following properties: (1) \textbf{Marginal Distribution}: The probability that $X_i=1$ is equal to $x_i$, i.e. $Pr [X_i=1] = x_i$ for each $i \in [t]$. (2) \textbf{Degree Preservation}: Sum of $X_i$'s should be equal to either $\floor{k}$ or $\ceil{k}$ with probability one, i.e. $Pr[\sum_{i \in [t]}X_i \in \{ \floor{k}, \ceil{k} \}] =1$. (3) \textbf{Negative Correlation}: For any $S \subseteq [t]$, (1) $Pr[ \land_{i \in S}X_i =0 ] \leq  \Pi_{i\in S}Pr[ X_i =0 ] $ (2) $Pr[ \land_{i \in S}X_i =1 ] \leq  \Pi_{i\in S}Pr[X_i =1]$. It follows that for any $x_i,x_j \in \vec{x}$, $\mathbb{E}[X_i=1|X_j=1] \leq x_i$.



Going back to the LPs (\ref{sys_obj},\ref{driver_obj},\ref{rider_obj}), we denote the optimal solutions to LP (\ref{sys_obj}), LP (\ref{driver_obj}), and LP (\ref{rider_obj}) by $\vec{x}^*$,$\vec{y}^*$ and $\vec{z}^*$ respectively. Further, we introduce the parameters $\alpha,\beta,\gamma \in [0,1]$ where $\alpha+ \beta+ \gamma \leq 1$ and each of these parameters decide the "weight" the algorithm places on each objective (the operator's profit, the offline group fairness, and the online group fairness objectives). We note that our algorithm makes use of the subroutine \PPDR \ (Probe with Permuted Dependent Rounding) \ shown in Algorithm \ref{alg:alg_sr}. 

\begin{algorithm}[h!]
   \caption{\PPDR($\vec{x}_v$) }
   \label{alg:alg_sr}
\begin{algorithmic}[1]
   \STATE  Apply dependent rounding to the fractional solution $\vec{x}_v$ to get a binary vector $\vec{X}_v$.
   \STATE  Choose a random permutation $\pi$ over the set $E_v$.
   \FOR{$i=1$ to $|E_v|$}
     \STATE Probe vertex $\pi(i)$ if it is available and $\vec{X}_v(\pi(i))=1$ 
      \IF{Probe is successful (i.e., a match)}
     \STATE break  
    \ENDIF
   \ENDFOR
\end{algorithmic}
\end{algorithm}

The procedure of our parameterized sampling algorithm $\tsf$ is shown in Algorithm \ref{alg:alg_intcap}. Specifically, when a vertex of type $v$ arrives at any time step we run \PPDR\big($\vec{x}_{v}^*$\big), \PPDR\big($\vec{y}_{v}^*$\big), or \PPDR\big($\vec{z}_{v}^*$\big) with probabilities $\alpha$, $\beta$, and $\gamma$, respectively. We do not run any of the \PPDR \ subroutines and instead reject the vertex with probability $1-(\alpha+\beta+\gamma)$. The LP constraint (\ref{constraint : patience}) guarantees that $\forall v \in V: \sum_{e \in E_r}s^{*}_{e} \leq \Delta_v$ where $\vec{s}^{*}$ could be $\vec{x}^{*},\vec{y}^{*},\text{or } \vec{z}^{*}$. Therefore, when \PPDR \ is invoked by the \textbf{degree preservation} property of dependent rounding the number of edges probed will not exceed $\Delta_v$, i.e. it would be within the patience limit. 

\begin{algorithm}[h!]
   \caption{$\tsf(\alpha,\beta,\gamma$) }
   \label{alg:alg_intcap}
\begin{algorithmic}[1]
   \STATE Let $v$ be the vertex type arriving at time $t$.
   \STATE With probability $\alpha$ run the subroutine, \PPDR\big($\vec{x}_{v}^*$\big). 
   \STATE With probability $\beta$ run the subroutine, \PPDR\big($\vec{y}_{v}^*$\big). 
   \STATE With probability $\gamma$ run the subroutine, \PPDR\big($\vec{z}_{v}^*$\big). 
   \STATE Reject the arriving vertex with probability $1- (\alpha+\beta+\gamma)$.
\end{algorithmic}
\end{algorithm}

Now we analyze $\tsf$ to prove Theorem \ref{main_th_kiid}. It would suffice to prove that for each edge $e$ the expected number of successful probes is at least $\alpha \frac{x^*_{e}}{2e}$,  $\beta \frac{y^*_{e}}{2e}$ and $\gamma \frac{z^*_{e}}{2e}$. And finally from the linearity of expectation we show that the worst case competitive ratio of the proposed online algorithm with parameters $\alpha,\beta$ and $\gamma$  is at least $(\frac{\alpha}{2e},\frac{\beta}{2e},\frac{\gamma}{2e})$ for the operator's profit and group fairness objectives on the offline and online sides of the matching, respectively.



A critical step is to lower bound the probability that a vertex $u$ is available (safe) at the beginning of round  $t \in [T]$. Let us denote the indicator random variable for that event by $SF_{u,t}$. The following lemma enables us to lower bound for the probability of $SF_{u,t}$.

\vspace{-0.2cm}
\begin{restatable}{lemma}{lemmatolerance}\label{lemma:available-tolerance}
$Pr[SF_{u,t}]  \geq \Big(1- \frac{t-1}{T}\Big)  \Big( 1-  \frac{1}{T}  \Big)^{t-1}  $. 
\end{restatable}


Now that we have established a lower bound on $Pr[SF_{u,t}]$, we lower bound the probability that an edge $e$ is probed by one of the \PPDR \ subroutines conditioned on the fact that $u$ is available (Lemma \ref{ref-lemma}). Let $1_{e,t}$ be the indicator that $e=(u,v)$ is probed by the $\tsf$ Algorithm at time $t$. Note that event $1_{e,t}$ occurs when (1) a vertex of type $v$ arrives at time $t$ and (2) $e$ is sampled by \PPDR($\vec{x_v}$), \PPDR($\vec{y_v}$), or \PPDR($\vec{z_v}$). 
\begin{restatable}{lemma}{lemmaref}\label{ref-lemma}
$ Pr[1_{e,t} \mid SF_{u,t} ] \geq \alpha \frac{x^*_{e}}{2T} $ ,$ Pr[1_{e,t} \mid SF_{u,t} ] \geq \beta \frac{y^*_{e}}{2T} $, $ Pr[1_{e,t} \mid SF_{u,t} ] \geq \gamma \frac{z^*_{e}}{2T} $ 
\end{restatable} 

Given the above lemmas Theorem \ref{main_th_kiid} can be proved. 

\subsection{Group Fairness for the \textbf{KAD} Setting:} \label{sec:kad}

For the \textbf{KAD} setting, the distribution over $V$ is time dependent and hence the probability of sampling a type $v$ in round $t$ is $\pvt \in [0,1]$ with $\sum_{v \in V} \pvt=1$. Further, we assume for the \textbf{KAD} setting that for every edge $e \in E$ we have $p_e=1$. This means that whenever an incoming vertex $v$ is assigned to a safe-to-add vertex $u$ the assignment is successful. This also means that any non-trivial values for the patience parameters $\pl$ and $\pr$ become meaningless and hence we can WLOG assume that $\forall u \in U, \forall v \in V, \pl=\pr=1$. From the above discussion, we have the following LP benchmarks for the operator's profit, the group fairness for the offline side and  the group fairness for the online side: 

\begin{align}
    & \textstyle\max \sum\limits_{t \in [T]} \sum\limits_{e \in E}{\wto \xet }  \label{sys_obj_kad} \\
    & \textstyle\max \min\limits_{g \in \grps}   \frac{\sum\limits_{t \in [T]} \sum\limits_{u \in U(g)} \sum\limits_{e \in E_u}\wtl \xet}{|U(g)|}  \label{driver_obj_kad} \\
    & \textstyle\max \min\limits_{g \in \grps} \frac{\sum\limits_{t \in [T]}\sum\limits_{v \in V(g)} \sum\limits_{e \in E_v} \wtr \xet}{\sum\limits_{v \in V(g)} n_v}  \label{rider_obj_kad}  
    \end{align}
    \vspace{-0.5cm}
    \begin{subequations}\label{constraintkad}
    \begin{align}
     &  \text{s.t} \quad  \forall e \in E, \forall t \in [T]:  0 \leq \xet \leq 1 \label{constraintkad : probkad} \\
    &  \textstyle\forall u \in U: \sum\limits_{t \in [T]} \sum\limits_{e \in E_u} \xet \leq 1  \label{constraintkad: lmatchkad} \\
    &  \textstyle\forall v \in V, \forall t \in [T]: \sum_{e \in E_v} \xet \leq \pvt  \label{constraintkad : rmatchkad} 
    \end{align}
   \end{subequations}

\begin{restatable}{lemma}{validkad}\label{valid_kad}
For the \textbf{KAD} setting, the optimal solutions of LP (\ref{sys_obj_kad}), LP (\ref{driver_obj_kad}) and LP (\ref{rider_obj_kad}) are upper bounds on the expected optimal that can be achieved by any algorithm for the operator's profit, the offline side group fairness objective, and the online side group fairness objective, respectively. 
\end{restatable}

Note that in the above LP we have $\xet$ as the probability for successfully assigning an edge in round $t$ (with an explicit dependence on $t$), unlike in the \textbf{KIID} setting where we had $x_e$ instead to denote the expected number of times edge $e$ is probed through all rounds. Similar to our solution for the \textbf{KIID} setting, we denote by $\xets$, $\yets$, and $\zets$ the optimal solutions of the LP benchmarks for the operator's profit, offline side group fairness, and online side group fairness, respectively. 

Having the optimal solutions to the LPs, we use algorithm $\tsfkad$ shown in Algorithm \ref{alg:alg_kad_group}. In $\tsfkad$ new parameters are introduced, specifically $\lambda$ and $\ret$ where $\ret$ is the probability that edge $e=(u,v)$ is safe to add in round $t$, i.e. the probability that $u$ is unmatched at the beginning of round $t$. For now we assume that we have the precise values of $\ret$ for all rounds and discuss how to obtain these values at the end of this subsection. Now conditioned on $v$ arriving at round $t$ and $e=(u,v)$ being safe to add, it follows that $e$ is sampled with probability $\alpha \frac{\xets}{\pvt}\frac{\lambda}{\ret} + \beta \frac{\yets}{\pvt}\frac{\lambda}{\ret}+ \gamma \frac{\zets}{\pvt}\frac{\lambda}{\ret}$ which would be a valid probability (positive and not exceeding 1) if $\ret \ge \lambda$. This follows from the fact that $\alpha,\beta,\gamma \in [0,1]$ and $\alpha+\beta+\gamma \leq 1$ and also by constraint (\ref{constraintkad : rmatchkad}) which leads to $\frac{\sum_{e \in E_v}\xet}{\pvt}\leq 1$. Further, if $\ret \ge \lambda$ then by constraint (\ref{constraintkad : rmatchkad}) we have $\sum_{e\in E_v} \Big( \alpha \frac{\xets}{\pvt}\frac{\lambda}{\ret} + \beta \frac{\yets}{\pvt}\frac{\lambda}{\ret}+ \gamma \frac{\zets}{\pvt}\frac{\lambda}{\ret} \Big) \leq 1$ and therefore the distribution is valid. Clearly, the value of $\lambda$ is important for the validity of the algorithm, the following lemma shows that $\lambda=\frac{1}{2}$ leads to a valid algorithm. 
\begin{restatable}{lemma}{kadalgvalid}\label{kad_alg_valid}
Algorithm $\tsfkad$ is valid for $\lambda=\frac{1}{2}$. 
\end{restatable}

\begin{algorithm}[h!]
   \caption{$\tsfkad(\alpha,\beta,\gamma$) }
   \label{alg:alg_kad_group}
\begin{algorithmic}[1]
   \STATE Let $v$ be the vertex type arriving at time $t$.
   \IF{ $E_{v,t}=\phi$} 
        \STATE Reject $v$ 
    \ELSE
       \STATE With probability $\alpha$ probe $e$ with probability $\frac{\xets}{\pvt}\frac{\lambda}{\ret}$. 
       \STATE With probability $\beta$ probe $e$ with probability $\frac{\yets}{\pvt}\frac{\lambda}{\ret}$. 
       \STATE With probability $\gamma$ probe $e$ with probability $\frac{\zets}{\pvt}\frac{\lambda}{\ret}$. 
       \STATE With probability $[1- (\alpha+\beta+\gamma)]$ reject $v$  .
    \ENDIF
\end{algorithmic}
\end{algorithm}


We now return to the issue of how to obtain the values of $\ret$ for all rounds. This can be done by using the simulation technique as done in \cite{dickerson2021allocation,adamczyk2015improved}. To elaborate, we note that we first solve the LPs (\ref{sys_obj_kad},\ref{driver_obj_kad},\ref{rider_obj_kad}) and hence have the values of $\xets$, $\yets$, and $\zets$. Now, for the first round $t=1$, clearly $\ret=1, \forall e\in E$. To obtain $\ret$ for $t=2$, we simulate the arrivals and algorithm a collection of times, and use the empirically estimated probability. More precisely, for $t=1$ we sample the arrival of vertex $v$ from $\pvt$ with $t=1$ ($\pvt$ values are given as part of the model), then we run our algorithm for the values of $\alpha,\beta,\gamma$ that we have chosen. Accordingly, at $t=2$ some vertex in $U$ might be matched. We do this simulation a number of times and then we take $\ret$ for $t=2$ to be the average of all runs. Now having the values of $\ret$ for $t=1$ and $t=2$, we further simulate the arrivals and the algorithm to obtain $\ret$ for $t=3$ and so on until we get $\ret$ for the last round $T$. We note that using the Chernoff bound \cite{mitzenmacher2017probability} we can rigorously characterize the error in this estimation, however by doing this simulation a number of times that is polynomial in the problem size, the error in the estimation would only affect the lower order terms in the competitive ration analysis \cite{dickerson2021allocation} and hence for simplicity it is ignored. Now, with Lemma \ref{kad_alg_valid} Theorem \ref{main_th_kad} can be proved (see Appendix (\ref{app:proofs})).

\subsection{Individual Fairness \textbf{KIID} and \textbf{KAD} Settings:}\label{sec:indiv}
For the case of Rawlsian (max-min) individual fairness, we consider each vertex of the offline side to belong to its own distinct group and the definition of group max-min fairness would lead to individual max-min fairness. On the other hand, for the online side a similar trick would not yield a meaningful criterion, we instead define the individual max-min fairness for the online side to equal $\min\limits_{t \in [T]}\E[\text{util}(v_t)] = \min\limits_{t \in [T]} \E[\sum_{e \in \mathcal{M}_{\rv_t}} \wtr)]$ where $\text{util}(v_t)$ is the utility received by the vertex arriving in round $t$. If we were to denote by $\xet$ the probability that the algorithm would successfully match $e$ in round $t$, then it follows straightforwardly that $\E[\text{util}(v_t)] = \sum_{e \in E_{v_t}} \wtr \xet$. We consider this definition to be the valid extension of max-min fairness for the online side as we are now concerned with the minimum utility across the online individuals (arriving vertices) which are $T$ many. The following lemma shows that we can solve two-sided individual max-min fairness by a reduction to two-sided group max-min fairness in the \textbf{KAD} arrival setting: 
\begin{restatable}{lemma}{reducgroup}\label{reduc_indiv_group}
Whether in the \textbf{KIID} or \textbf{KAD} setting, a given instance of two-sided individual max-min fairness can be converted to an instance of two-sided group max-min fairness in the \textbf{KAD} setting. 
\end{restatable}
The above Lemma with algorithm $\tsfkad$ can be used to prove Theorem \ref{main_th_indiv} as shown in Appendix (\ref{app:proofs}).

\section{Experiments}\label{sec:experiments}

\begin{figure*}[t!]
  \centering
  \includegraphics[width=0.8\textwidth]{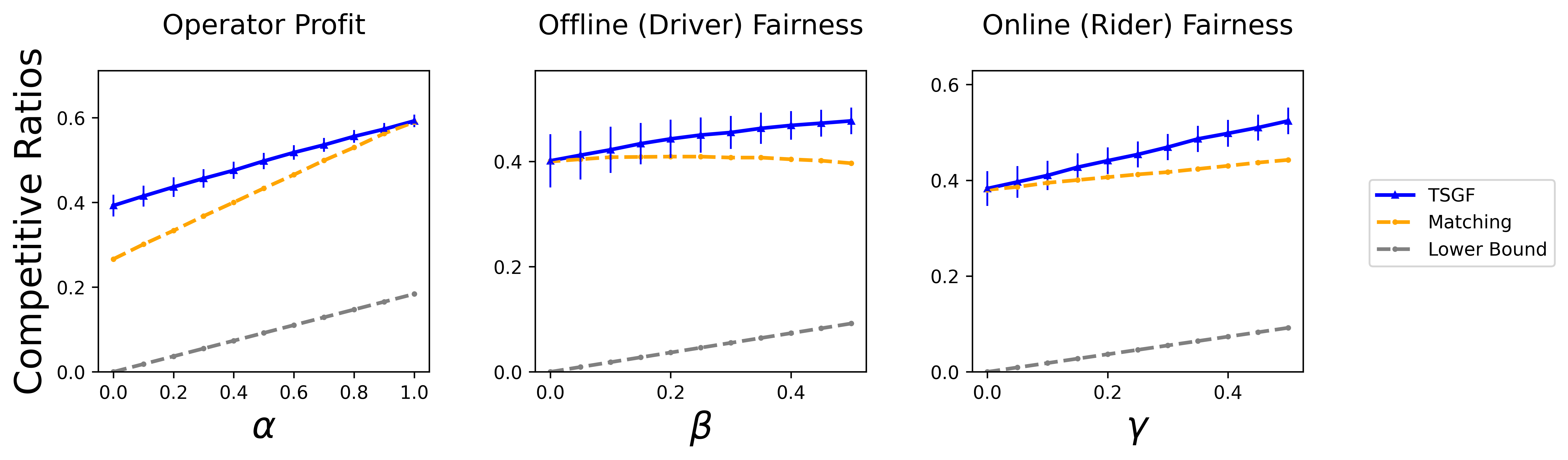}
  \caption{Competitive ratios for $\tsf$ over the operator's profit, offline (driver) fairness objective, and online (rider) fairness objective with different values of $\alpha,\beta,\gamma$. Note that ``Matching'' refers to the case where driver and rider utilities are set to 1 across all edges. The experiment is run with $\alpha = \{0,0.1,0.2,...,1\}$, and $\beta = \gamma = \frac{1-\alpha}{2}$. Higher competitive ratio indicates better performance.
}
\vspace{-0.4cm}
  \label{fig:fm_exp_fig}
\end{figure*}

In this section, we verify the performance of our algorithm and our theoretical lower bounds for the \text{KIID} and group fairness setting using algorithm $\tsf$ (Section \ref{sec:kiid}). We note that none of the previous work consider our three-sided setting. We use rideshare as an application example of online bipartite matching~\cite[see also, e.g.,][]{dickerson2021allocation,nanda2019balancing,xutrade, barann2017open}. We expect similar results and performance to hold in other matching applications such as crowdsourcing. 

\paragraph{\textbf{Experimental Setup}:} As done in previous work, the drivers' side is the offline (static) side whereas the riders' side is the online side. We run our experiments over the widely used New York City (NYC)\ yellow cabs dataset~\cite{sekulic2021spatially, nanda2019balancing,xutrade, AlonsoPredictive2017} which contains records of taxi trips in the NYC area from 2013. Each record contains a unique (anonymized) ID of the driver, the coordinates of start and end locations of the trip, distance of the trip, and  additional metadata.

Similar to \cite{dickerson2021allocation,nanda2019balancing}, we bin the starting and ending latitudes and longitudes by dividing the latitudes from $40.4^{\circ}$ to $40.95^{\circ}$ and longitudes from $-73^{\circ}$ to $-75^{\circ}$ into equally spaced grids of step size $0.005$. This enables us to define each driver and request type based on its starting and ending bins. We pick out the trips between 7pm and 8pm on January 31, 2013, which is a rush hour with 10,814 drivers and 35,109 trips. We set driver patience $\Delta_u$ to 3. Following \cite{xutrade}, we uniformly sample rider patience $\Delta_v$ from $\{1, 2\}$.




Since the dataset does not include demographic information, for each vertex we randomly sample the group membership  \cite{nanda2019balancing}. Specifically, we randomly assign $70\%$ of the riders and drivers to be advantaged and the rest to be disadvantaged. The value of $p_e$ for $e=(u,v)$ depends on whether the vertices belong to the advantaged or disadvantaged group. Specifically, $p_e=0.6$ if both vertices are advantaged, $p_e=0.3$ if both are disadvantaged, and $p_e=0.1$ for other cases.

In addition to this, a key component of our work is the use of driver and rider specific utilities. We follow the work of~\cite{suhr2019two} to set the utilities. We adopt the Manhattan distance metric rather than the Euclidean distance metric since the former is a better proxy for length of taxi trips in New York City. We set the operator's utility to the rider's trip length $w^O_e=\text{tripLength$(v)$}$---a rough proxy for profit. In addition, the rider's utility over an edge $e=(u,v)$ is set to $w^V_e=-\text{dist$(u,v)$}$ where $\text{dist$(u,v)$}$ is the distance between the rider and the driver. The driver's utility is set to $w^U_e=\text{tripLength$(v)$}-\text{dist$(u,v)$}$. Whereas the trip length $\text{tripLength$(v)$}$ is available in the dataset, the distance between the rider and the driver $\text{dist$(u,v)$}$ is not. We therefore simulate the distance, by creating an equally spaced grid with step size $0.005$ around the starting coordinates of the trip. This results in 81 possible coordinates in the vicinity of the starting coordinates of the trip. We then randomly choose one of these 81 coordinates to be the location of the driver when the trip was requested. Then $dist(u, v)$ is the distance between this coordinate to the start coordinate of the trip. This is a valid approximation since the platform would not assign drivers unreasonably far away to pickup a rider. Lastly, we scale the utilities by a constant to prevent them from being negative. 



We run $\tsf$ at the scale of $|U|=49$, $|V|=172$ for 100 trials. During each trial, we randomly sample 49 drivers and 172 requests between 7 and 8pm, and run $\tsf$ 100 times to measure the expected competitive ratios of this trial. We then averaged the competitive ratios over all trials, and the results are reported in figure \ref{fig:fm_exp_fig}. Code to reproduce our experiments is available in the blinded format\footnote{https://github.com/anonymousUser634534/TSGF}; we will release that code in deblinded form upon acceptance.

\vspace{-0.cm}
\paragraph{\textbf{Performance of $\tsf$ with Varied Parameters}:} Figure~\ref{fig:fm_exp_fig} shows the performance of our algorithm over the three objectives: operator's profit, offline (driver) group fairness, and online (rider) group fairness. It is clear that the algorithm behaves as expected with all objectives being steadily above their theoretical lower bound. More importantly, we see that increasing the weight for an objective leads to better performance for that objective. I.e., a higher weight for $\beta$ leads to better performance for the offline side fairness and similar observations follow in the case of $\alpha$ for the operator's objective and in the case of $\gamma$ for the online-fairness. This also indicates the limitation in previous work which only considered fairness for one-side since their algorithms would not be able to improve the fairness for the other ignored side. 

Furthermore, previous work \citep[e.g.,][]{nanda2019balancing,xutrade,ma2022group} only considered the matching size when optimizing the fairness objective for the offline (drivers) or online (riders) side. This is in contrast to our setting where we consider the matching quality. To see the effect of ignoring the matching quality and only considering the size, we run the same experiments with $w^U_e=w^V_e=1, \forall e \in E$, i.e. the quality is ignored. The results are shown shown in the graph labelled ``Matching'' in figure~\ref{fig:fm_exp_fig}, it is clear that ignoring the match quality leads to noticeably worse results. 

\paragraph{\textbf{Comparison to Heuristics}:} We also compare the performance of $\tsf$ against three other heuristics. In particular, we consider Greedy-O which is a greedy algorithm that upon the arrival of an online vertex (rider) $v$ picks the edge $e \in E_v$ with maximum value of $p_e w^O_e$ until it either results in a match or the patience quota is reached. We also consider Greedy-R which is identical to Greedy-O except that it greedily picks the edge with maximum value of $p_e w^V_e$ instead, therefore maximizing the rider's utility in a greedy fashion. Moreover, we consider Greedy-D which is a greedy algorithm that upon the arrival of an online vertex $v$, first finds the group on the offline side with the lowest average utility so far, then it greedily picks an offline vertex (driver) $u \in E_v$ from this group (if possible) which has the maximum utility until it either results in a match or the patience limit is reached. We carried out 100 trials to compare the performance of $\tsf$ with the greedy algorithms, where each trial contains 49 randomly sampled drivers and 172 requests and is repeated 100 times. The aggregated results are displayed in table \ref{table:heuristic_comparison}. We see that $\tsf$ outperforms the heuristics with the exception of a small under-performance in comparison to Greedy-D. However, using Greedy-D we cannot tune the weights ($\alpha$, $\beta$, and $\gamma$) to balance the objectives as we can in the case of $\tsf$.

\begin{table}
\centering
\begin{tabular}{@{}lccc@{}}
\toprule
                               & \textbf{Profit} & \textbf{\begin{tabular}[c]{@{}c@{}}Driver \\ Fairness\end{tabular}} & \textbf{\begin{tabular}[c]{@{}c@{}}Rider \\ Fairness\end{tabular}} \\ \midrule
Greedy-O              & 0.431           & 0.549                                                              & 0.503                                                               \\
\textbf{$\tsf$ ($\alpha = 1$)} & 0.595           & 0.398                                                              & 0.384                                                               \\ \midrule
Greedy-D              & 0.371           & 0.609                                                              & 0.563                                                               \\
\textbf{$\tsf$ ($\beta = 1$)}  & 0.517           & 0.571                                                              & 0.44                                                                \\ \midrule
Greedy-R              & 0.316           & 0.504                                                              & 0.513                                                               \\
\textbf{$\tsf$ ($\gamma = 1$)} & 0.252           & 0.353                                                              & 0.574                                                               \\ \bottomrule

\end{tabular}
\caption{Competitive ratios of $\tsf$ with Greedy heuristics on the NYC dataset at $|U| = 49$, $|V| = 172$. Higher competitive ratio indicates better performance. }
\label{table:heuristic_comparison}
\end{table}
\vspace{-0.1cm}


\section*{Acknowledgments}
This research was supported in part by 
NSF CAREER Award IIS-1846237,  
NSF Award CCF-1918749,  
NSF Award CCF-1852352,  
NSF Award SMA-2039862,  
NIST MSE Award \#20126334,  
DARPA GARD \#HR00112020007,  
DARPA SI3-CMD \#S4761,  
DoD WHS Award \#HQ003420F0035,  
ARPA-E DIFFERENTIATE Award \#1257037,  
ARL Award W911NF2120076, 
and gifts by research awards from
Amazon, 
and Google.  
We are grateful to Pan Xu for advice and comments on earlier versions of this work.



\bibliography{refs}
\clearpage
\newpage 
\appendix
\section{Online Matching Model Details}\label{app:model_details}
\subsection{Arrival Setting (\textbf{KIID} and \textbf{KAD}):}
The modelling choices we have made follow standard settings in online matching \cite{mehta2013online,alaei2013online}. To elaborate further, the initial seminal paper on online matching \cite{karp1990optimal} does not assume any prior knowledge on the arrival of the online vertices of $V$ and follows adversarial analysis to establish theoretical guarantees on the competitive ratio. In addition to overly pessimistic theoretical results, the lack of prior knowledge is often an unrealistic assumption. Most decision makers in online matching settings are able to gain knowledge on the arrival rates of the online vertices and this knowledge can be used to build more realistic probabilistic knowledge of the arrival. 

Specifically, the Known Independent and Identically Distributed \textbf{KIID} model is an established model in online matching \cite{feldman2009online,mehta2013online,bahmani2010improved,manshadi2012online,dickerson2019balancing}. In this model, the collection of arriving vertices on the online side belong to a finite set of known types where the type of a vertex $v$ decides the edge connections $E_v$ it has to the vertices of $U$ along with the weights $w_{e}, \forall e\in E_v$ of those edges. Further, a given vertex of type $v$ arrives with the same probability $p_v$ in every round. These arrival probabilities can be estimated easily from historical data based on previous matchings. 

While the \textbf{KIID} model utilizes prior knowledge which is frequently available in practical applications, it is still restrictive since it assumes that the probabilities do not vary through time. The Known Adversarial Arrival \textbf{KAD} model (also known as prophet inequality matching) on the other hand, takes into account the dynamic variation in the probabilities. Therefore, the probability a vertex of type $v$ arrives in round $t$ is $p_{v,t}$ instead of being constant for every round $t$. This model is also well-established in the matching literature and has been used in a collection of papers such as \cite{alaei2012online,brubach2016new,dickerson2021allocation,dickerson2019online}. Despite the fact that the \textbf{KAD} model is well-motivated and richer than the \textbf{KIID} model it was not used in the one-sided online fair matching papers of \cite{nanda2019balancing,xutrade}.

\subsection{Patience:}
The patience parameter of a vertex $\Delta_u$ (or $\Delta_v$) for an offline vertex $u$ (or an online vertex $v$) models its tolerance for unsuccessful probes (match attempts) before leaving the system. We note that this is an important detail in the online matching model since it is frequently the case that the vertices in the online matching applications (such as advertising, crowdsourcing, and ridesharing) represent human participants who would only tolerate a fixed number of failed matching attempts before leaving the system. Like the \textbf{KIID} and \textbf{KAD} arrival models, the patience parameter is also well-established in online matching, see for example \cite{mehta2013online,bansal2010lp,adamczyk2015improved}. Despite the importance of this parameter, the previous work in fair online matching did not consider the patience issue for both sides simultaneously \cite{nanda2019balancing,xutrade}, handling both parameters at the same time is more challenging and leads to more tedious derivations. 

We further elaborate on the meaning of the patience for both the online and offline sides, we note again that this is following the research literature on online matching: 
\subsubsection{Offline Patience:}
Consider a vertex $u$ with patience $\Delta_u$, then vertex $u$ will remain on the offline side $U$ unless it is successfully matched or it receives $\Delta_u$ many failed matching attempts. As a concrete example, consider a vertex $u_1$ with patience $\Delta_{u_1}=2$. Clearly, in the first round ($t=1$) $u_1$ will be in the offline side $U$, suppose an unsuccessful matching attempt (unsuccessful probe) is made in this round, then in the next round $u_1$ will still be there. Suppose that the next round when $u_1$ is probed is in the fifth round ($t=5$), then if the probe is successful then $u_1$ is matched and will be removed from the offline side in the next rounds ($t>5$), but also if the match is unsuccessful then $u_1$ will not be matched but will still be removed for all of the next rounds ($t>5$) since it has a patience $\Delta_{u_1}=2$ and therefore can only take two failed matching attempts before leaving. 

\subsubsection{Online Patience:}
Unlike the offline side, an online vertex $v$ would arrive in a round $t$ and must be matched or rejected in that given round. While in a round $t$ we can at most match one online vertex (which is the arriving vertex $v$) to some offline vertex $u$, we can make multiple match attempts (probes) from $v$ to the vertices it is connected to in $U$ in that round $t$. The patience $\Delta_v$ of $v$ decides the upper limit on the number of failed attempts we can make in round $t$ before $v$ leaves the system and can no longer be matched even if a possible match was still not attempted. As a concrete example, suppose vertex of type $v_1$ with $\Delta_{v_1}=3$ arrives in round $t=7$ and that $v_1$ is connected to a total of four vertices $\{u_1,u_2,u_3,u_4\}$ in $U$ all of which are still available (i.e. unmatched and still have not ran out of patience), suppose we make match attempts (probes) to $u_1$ then $u_2$ then $u_3$, it follows since $\Delta_{v_1}=3$ that $v_1$ has left the system and we can no longer even attempt to match it to $u_4$ despite that fact that its available. Further, if at any probe attempt $v_1$ was matched then no further probe attempts are made to $v_1$, e.g. if the first probe ($v_1$ to $u_1$) in the above discussion was successful, then $v_1$ and $u_1$ are matched to each other and we cannot attempt to match $v_1$  to $u_2,u_3, \text{or }u_4$. 


\section{Proofs}\label{app:proofs}
Here we include the missing proofs. Each lemma/theorem is restated followed by its proof.
\subsection{Proofs for Section \ref{sec:kiid}}
\validkiid*
\begin{proof}
We follow a similar proof to that used in \cite{bansal2010lp}. We shall focus on the operator's profit objective since the other objectives follow by very similar arguments. First, we note that LP(\ref{sys_obj}) uses the expected values of the problem parameters, i.e. if we consider a specific graph realization $G$, then let $N^G_v$ be the number of arrival for vertex type $v$, then it follows that LP(\ref{sys_obj}) uses the expected values since $\E_{\mathcal{I}}[N^G_v]=1, \forall v \in V$ where $\E_{\mathcal{I}}[.]$ is an expectation over the randomness of the instance. We shall therefore refer to the value of LP(\ref{sys_obj}) as $LP(\EI[G])$.

To prove that $LP(\EI(G))$ is a valid upper bound it suffices to show that $LP(\EI[G[) \ge \EI[LP(G)]$ where $LP(G)$ is the optimal LP value of a realized instance $G$ and $\EI[LP(G)]$ is the expected value of that optimal LP value. Let us then consider a specific realization $G'$, its corresponding LP would be the following: 
{\small
\begin{align}
    & \textstyle\max \sum_{e' \in E'}{w^O_{e'} p_{e'} x_{e'}}  \label{sys_obj_real} 
    \end{align}
    \vspace{-0.5cm}
    \begin{subequations}\label{constraint_real}
    \begin{align}
    &  \text{s.t} \quad  \forall e' \in E':  0 \leq x_{e'} \leq 1 \label{constraint_real : assignment_real } \\
    &  \textstyle\forall u \in U: \sum_{e' \in E'_u} x_{e'} p_{e'} \leq 1  \label{constraint_real: capacity_real} \\
    &  \textstyle\forall u \in U: \sum_{e' \in E'_u} x_{e'} \leq \Delta_u  \label{constraint_real: cancellation_real} \\ 
    &  \textstyle\forall v' \in V': \sum_{e' \in E'_{v'}} x_{e'} p_{e'} \leq 1  \label{constraint_real : arrivalrate_real} \\
    &  \textstyle\forall v' \in V':  \sum_{e' \in E'_{v'}} x_{e'}  \leq \Delta_{v'} \label{constraint_real : patience_real} 
    \end{align}
   \end{subequations}
}
where $V'$ is the realization of the online side. It is clear that for a given realization $G'=(U,V',E')$ the above LP(\ref{sys_obj_real}) is an upper bound on the operator's objective value for that realization. 

Now we prove that $LP(\EI[G]) \ge \EI[LP(G)]$. The dual of the LP for the realization $G'$ is the following: 
{\small
\begin{align}
    & \textstyle\min \sum_{u \in U} (\alpha_u +\pl \beta_u) +  \sum_{v' \in V'} (\alpha_{v'} +\Delta_{v'} \beta_{v'}) + \sum_{(u,v')} \gamma_{u,v'} \label{dual_sys_obj_real}  
    \end{align}
    \vspace{-0.5cm}
    \begin{subequations}\label{dual_constraint_real}
    \begin{align}
        \notag & \text{s.t.} \quad  \forall u \in U, \forall v' \in V':  \\
      & \beta_u + \beta_{v'} + p_{(u,v')}(\alpha_u + \alpha_{v'}) + \gamma_{(u,v')}  \ge w^O_{(u,v')}  p_{(u,v')}  \label{dual_constraint_real : dual_real_1} \\
     & \alpha_u, \alpha_{v'}, \beta_u, \beta_{v'}, \gamma_{(u,v')} \ge 0  \label{dual_constraint_real: dual_real_2} 
    \end{align}
   \end{subequations}
}
Consider the graph with the expected number of arrival $\EI(G)$ it would have a dual of the above form, let $\vec{\alpha}^*,\vec{\beta}^*,\vec{\gamma}^*$ be the optimal solution of its corresponding dual. Then it follows by the strong duality of LPs that solution $\vec{\alpha}^*,\vec{\beta}^*,\vec{\gamma}^*$ would have a value of $LP(\EI[G])$. Now for the instance $G'$, we shall use the following dual solution $\vec{\ah},\vec{\bh},\vec{\gh}$ which is set as follows: 
\begin{itemize}
    \item $\forall u \in U: \ah_u= \alpha^*_u, \bh_u = \alpha^*_u$. 
    \item $\forall v' \in V'$ of type $v$: $\ah_{v'}= \alpha^*_v, \bh_{v'}= \beta^*_v$. 
    \item $\forall u \in U, \forall v' \in V'$ of type $v$: $\gh_{(u,v')}= \gamma^*_{(u,v)}$. 
\end{itemize}
Note that the new solution $\vec{\ah},\vec{\bh},\vec{\gh}$ is a feasible dual solution since it satisfies constraints \ref{dual_constraint_real : dual_real_1} and \ref{dual_constraint_real: dual_real_2}. By weak duality the value of the solution $\vec{\ah},\vec{\bh},\vec{\gh}$ upper bounds $LP(G')$. Now if we were to denote the number of vertices of type $v$ that arrived in instance $G'$ by $n^{G'}_v$, then the value of the solution $\vec{\ah},\vec{\bh},\vec{\gh}$ satisfies:
\begin{align*}
    & \sum_{u \in U} (\ah_u +\pl \bh_u) +  \sum_{v' \in V'} (\ah_{v'} +\Delta_{v'} \bh_{v'}) + \sum_{(u,v')} \gh_{u,v'} \\
    & =\sum_{u \in U} (\alpha^*_u +\pl \beta^*_u) +  \sum_{v \in V} n^{G'}_v (\alpha^*_{v} +\Delta_{v} \beta^*_{v}) + \sum_{(u,v)} n^{G'}_v \gamma^*_{u,v}\\
    & \ge LP(G')
\end{align*}
Now taking the expectation, we get: 
\begin{align*}
    & \EI[LP(G')]  \\
    & \leq \EI\Big[\sum_{u \in U} (\ah_u +\pl \bh_u) +  \sum_{v' \in V'} (\ah_{v'} +\Delta_{v'} \bh_{v'}) + \sum_{(u,v')} \gh_{u,v'} \Big] \\ 
    & = \EI\Big[\sum_{u \in U} (\alpha^*_u +\pl \beta^*_u) +  \sum_{v \in V} n^{G'}_v (\alpha^*_{v} +\Delta_{v} \beta^*_{v}) + \sum_{(u,v)} n^{G'}_v \gamma^*_{u,v} \Big] \\
    & = \sum_{u \in U} (\alpha^*_u +\pl \beta^*_u) +  \sum_{v \in V} \EI[n^{G'}_v] (\alpha^*_{v} +\Delta_{v} \beta^*_{v}) + \sum_{(u,v)} \EI[n^{G'}_v] \gamma^*_{u,v} \\ & = \sum_{u \in U} (\alpha^*_u +\pl \beta^*_u) +  \sum_{v \in V} (\alpha^*_{v} +\Delta_{v} \beta^*_{v}) + \sum_{(u,v)} \gamma^*_{u,v} \\
    & = LP(\EI[G])
\end{align*}
For the offline and online group fairness objectives, we use the same steps. The difference would be in the constraints of the dual program, however following a similar assignment as done from $\vec{\alpha}^*,\vec{\beta}^*,\vec{\gamma}^*$ to $\vec{\ah},\vec{\bh},\vec{\gh}$ is sufficient to prove the lemma for both fairness objectives. 
\end{proof}

Before we prove Lemma \ref{lemma:available-tolerance} for the lower bound on the probability of $SF_{u,t}$. We have to first introduce the following two lemmas. Specifically, let $A_{u,t}$ be the number of successful assignments that $u$ received and accepted before round $t$. Then the following lemma holds.  
\begin{restatable}{lemma}{safe} \label{lemma:safe}
For any given vertex $u$ at time $t \in [T]$ , $P[A_{u,t} =0] \geq  \Big( 1-  \frac{1}{T}  \Big)^{t-1}$. 
\end{restatable} 
\begin{proof}
Let $X_{e,k}$ be the indicator random variable for $u$ receiving an arrival request of type $v$ where $e \in E_u$  and $k <t$. Let $Y_{e,k}$ be the indicator random variable that the edge $e$ gets sampled by the $\tsf(\alpha,\beta,\gamma)$ algorithm at time $k <t$. Let $Z_{e,k}$  be the indicator random variable that assignment $e=(u,v)$ is successful (a match) at time $k <t$. Then $ A_{u,t}=\sum_{k<t} \sum_{e \in E_u} X_{e,k}Y_{e,k}Z_{e,k}$. 
\vspace*{-2mm}
{\small
\begin{align*}
& Pr[A_{u,t} = 0] = \Pi_{k<t} Pr\Big[\sum_{e =(u,v) \in E_u} X_{e,k}Y_{e,k}Z_{e,k} = 0 \Big] \\
& = \Pi_{k<t}\Big(1- Pr\Big[\sum_{e \in E_u} X_{e,k}Y_{e,k}Z_{e,k} \geq 1 \Big]\Big) \\
& \ge \Pi_{k<t}\Big( 1- \sum_{e \in E_u} \frac{1}{T} \cdot  \big( \alpha x_{e}^* +  \beta \frac{y_{e}^*}{q_{v}} + \gamma \frac{z_{e}^*}{q_{v}} \big)\cdot p_{e}  \Big) \\
& =  \Pi_{k<t}\Big( 1-  \frac{1}{T} \cdot  \big( \alpha \sum_{e \in E_u} {x_{e}^*p_{e}  }+  \beta \sum_{e \in E_u} {y_{e}^*p_{e}  } + \gamma \sum_{e \in E_u} {z_{e}^*p_{e}  }\big) \Big) \\
& \geq  \Pi_{k<t}\Big( 1- \frac{1}{T} \cdot  \big( \alpha \cdot 1 +  \beta \cdot 1 + \gamma \cdot 1   \big) \Big) \\
& \geq  \Pi_{k<t}\Big( 1-  \frac{1}{T}  \Big) = \Big( 1-  \frac{1}{T}  \Big)^{t-1} 
\end{align*}
}
\vspace*{-2mm}
\end{proof} 

Now we lower bound the probability that $u$ was probed less than $\Delta_u$ times prior to $t$. Denote the number of probes received by $u$ before $t$ by $B_{u,t}$, then the following lemma holds: 
\begin{restatable}{lemma}{patiencevio} \label{lemma:patience-violation}
$Pr[B_{u,t} < \Delta_u]  \geq 1- \frac{t-1}{T} $. 
\end{restatable}
\begin{proof}
First it is clear that $B_{u,t}=\sum_{k<t}\sum_{e \in E_u}X_{e,k}Y_{e,k}$. 

{
\begin{align*}
    & \mathbb{E}[B_{u,t}] = \sum_{k<t}\sum_{e \in E_u} \mathbb{E}[X_{e,k}Y_{e,k} ] \\ 
    & \leq \sum_{k<t} \sum_{e \in E_u} \frac{1}{T}  \Big( \alpha x_{e}^* + \beta y_{e}^* + \gamma z_{e}^* \Big) \\ 
    & \leq \sum_{k<t} \frac{1}{T} \Big( \alpha \sum_{e \in E_d} {x_{e}^*  }+  \beta \sum_{e \in E_u} {y_{e}^*} +\gamma \sum_{e \in E_u} z_{e}^* \Big) \\ 
    & \leq  \sum_{k<t} \frac{\Delta_u}{T} (\alpha +  \beta+\gamma ) \leq \frac{(t-1)\Delta_u}{T}
\end{align*}
}
The inequality before the last follows from $ (\alpha + \beta +\gamma) \leq 1$. Now using Markov's inequality: $Pr[B_{u,t} < \Delta_u]  \geq  1- \frac{\mathbb{E}[B_{u,t}] }{\Delta_u}$, we get $\implies Pr[B_{u,t} < \Delta_u]  \geq  1- \frac{t-1}{T}$. 
\end{proof}
Now we restate Lemma \ref{lemma:available-tolerance} and prove it.  
\lemmatolerance*
\begin{proof}
Consider a given edge $e \in E_u$ where $k < t$
\begin{align*}
    & \mathbb{E}[X_{e,k}Y_{e,k} \mid A_{u,t} =0 ] =  \mathbb{E}[X_{e,k}Y_{e,k} \mid A_{u,k} =0] \\ 
    & = \frac{Pr[X_{e,k}=1 , Y_{e,k} =1 ,  Z_{e,k} =0 ]}{Pr[A_{u,k}=0]} \\
    & \leq  \frac{\frac{1}{T} \cdot  \big( \alpha x_{e}^* +  \beta y_e^* + \gamma z_{e}^*  \big)\cdot (1-p_e) }{1- \sum_{e \in E_d} \frac{1}{T} \cdot  \big( \alpha x_{e}^* +  \beta y_e^* + \gamma z_{e}^*  \big)\cdot p_{e}} \\
    & =  \frac{\frac{1}{T} \cdot  \big( \alpha x_{e}^* +  \beta y_e^* + \gamma z_{e}^*  \big)\cdot (1-p_e) }{1 - p_e + p_e \Big(1 - \sum_{e \in E_d} \frac{1}{T} \cdot  \big( \alpha x_{e}^* +  \beta y_e^* + \gamma z_{e}^*  \big) \Big) } \\
    & \leq  {\frac{1}{T} \cdot  \big( \alpha x_{e}^* +  \beta y_e^* + \gamma z_{e}^*  \big)\cdot  } 
\end{align*}
The above inequality is due to the fact that $\sum_{e \in E_u} \frac{1}{T} \big( \alpha {x_{e}^*}+  \beta {y_e^*} + \gamma{z_{e}^*} \big) \leq \frac{\pl}{T} < 1$. 
{
\begin{align*}
&  \mathbb{E}[B_{u,t} | A_{u,t} =0] = \sum_{k <t}\sum_{e \in E_u}  \mathbb{E}[X_{e,k}Y_{e,k} | A_{u,k} =0 ]     \\
   & \leq \sum_{k<t} \sum_{e \in E_u} \frac{1}{T}  \Big( \alpha x_{e}^* + \beta y_{e}^* + \gamma z_{e}^* \Big) \\ 
    & \leq \sum_{k<t} \frac{1}{T} \Big( \alpha \sum_{e \in E_u} {x_{e}^*  }+  \beta \sum_{e \in E_u} {y_{e}^*} +\gamma \sum_{e \in E_u} z_{e}^* \Big)    \\ 
    &  \leq \sum_{k<t}  \frac{1}{T} \Big( \alpha \cdot \Delta_u +  \beta \cdot \Delta_d+\gamma \cdot \Delta_u \Big) \\
    & =  \sum_{k<t} \frac{\Delta_u}{T} (\alpha +  \beta+\gamma ) \leq \frac{(t-1)\Delta_u}{T}
\end{align*}
}
Therefore the expected number of assignments (probes) to vertex $u$ until time $t$ is at most $\frac{(t-1)\Delta_u}{T}$. Therefore, we have:
\begin{align*}
     & Pr[B_{u,t} < \Delta_u |A_{u,t} =0] \geq  1- \frac{\mathbb{E}[B_{u,t} |A_{u,t} =0] }{\Delta_d}  \\ 
    &\geq 1  - \frac{(t-1)  \Delta_u}{T \Delta_u} \geq 1- \frac{t-1}{T} 
\end{align*}
It is to be noted that $B_{u,t}$ is the total number of probes $u$ received before round $t$. Thus, we have that the events $(B_{u,t} < \Delta_u)$ and $(A_{u,t} = 0)$ are positively correlated. Therefore, 
\begin{align*}
& Pr[SF_{u,t}]  \geq Pr[(B_{u,t} < \Delta_u) \land (A_{u,t} =0) ] \\
& \geq Pr[B_{u,t}<\Delta_d|A_{u,t} =0] Pr[A_{u,t} =0] \\
& Pr[SF_{u,t}] \geq \Big(1- \frac{t-1}{T}\Big)  \Big( 1-  \frac{1}{T}  \Big)^{t-1}
\end{align*}
\end{proof}. 
\lemmaref*
\begin{proof}
In this part we prove that the probability that edge $e$ is probed at time $t$ is at least $\alpha \frac{x^*_{e}}{2T}$. Note that the probability that a vertex of type $v$ arrives at time $t$ and that Algorithm \ref{alg:alg_intcap} calls the subroutine \PPDR($\vec{x_r}$) is $\alpha \frac{1}{T}$. Let $E_{v,\bar{e}}$ be the set of edges in $E_v$ excluding $e=(u,v)$. For each edge $e' \in E_{v,\bar{e}}$ let $Y_{e'}$ be the indicator for $e'$ being before $e$ in the random order of $\pi$ (in algorithm \ref{alg:alg_sr}) and let $Z_{e'}$ be the probability that the assignment is successful for $e'$. It is clear that $\E[Y_{e'}] = 1/2$ and that $\E[Z_{e'}]=p_{e'}$. Now we have: 
\begin{align}
&   Pr[1_{e,t} \mid SF_{u,t} ]  \\
&  \geq \alpha \frac{1}{T} Pr[X_{e}=1] Pr \big[\sum_{e' \in E_{r,\bar{e}}}X_{e'}Y_{e'} Z_{e'} \mid X_{e}=1\big] \\
& = \alpha \frac{Pr[X_{e}=1]}{T}  \big( 1 - Pr \big[\sum_{e' \in E_{v,\bar{e}}}X_{e'} Y_{e'} Z_{e'}\geq 1 \mid X_{e}=1\big] \big) \\
&  \geq \alpha \frac{Pr[X_{e}=1]}{T} \big( 1 - \mathbb{E} \big[\sum_{e' \in E_{v,\bar{e}}}X_{e'} Y_{e'} Z_{e'}\geq 1 \mid X_{e}=1\big] \big) \label{markovinequality}  \\
&\geq \alpha \frac{Pr[X_{e}=1]}{T} \big( 1 - \sum_{e' \in E_{v,\bar{e}}}\mathbb{E} \big[X_{e'} Y_{e'} Z_{e'} \geq 1 \mid X_{e}=1\big] \big) \label{loe} \\
& \geq \alpha \frac{x^*_{e}}{T} \big( 1 - \sum_{e' \in E_{v,\bar{e}} }  x_{e'}^*  \frac{1}{2}p_{e'}  \big) \label{dependentrounding}\\
& \geq \alpha \frac{x^*_{e}}{T} \big( 1 - \frac{1}{2} \big)  = \alpha \frac{x^*_{e}}{2T}  \label{constraint}
\end{align}
Applying Markov inequality we get the inequality (\ref{markovinequality}). By linearity of expectation we get inequality (\ref{loe}). Since $X_{e}$ and $X_{e'}$ are negatively correlated to each other from the  Negative Correlation property of Dependent Rounding we have $\E[X_{e'} \mid X_{e}=1] \leq {x}_{e}^*$ and we get (\ref{dependentrounding}). The last inequality (\ref{constraint}) is due the fact that for any feasible solution $\{x_{e}^*\}$ the constraints imply that $\sum_{e \in E_v}x_{e}^*p_{e} \leq 1$ for all $v \in V$. Using similar analysis we can also prove that $Pr[1_{e,t} \mid SF_{u,t} ] \geq \beta \frac{y^*_{e}}{2T}$ and $ Pr[1_{e,t} \mid SF_{u,t} ] \geq \gamma \frac{z^*_{e}}{2T}$.
\end{proof}

Now we restate and prove Theorem \ref{main_th_kiid}.
\mainthkiid*
\begin{proof}
Denote the expected number of probes on each edge $e \in E$ resulting from \PPDR\big($\vec{x}_{v}^*$\big) by $n^{x}_{e}$.
It follows that:
{
\begin{align*}
    & n^{x}_{e} \geq \sum_{t=1}^{T} Pr[1_{e,t}] = \sum_{t=1}^{T} Pr[1_{e,t} \mid SF_{u,t} ] Pr[ SF_{u,t} ] \\
    & \geq \sum_{t=1}^{T}  \Big( 1-  \frac{1}{T}  \Big)^{t-1}   \Big( 1-  \frac{t-1}{T}  \Big)\Big( \alpha \frac{x^*_{e}}{2T}  \Big) \xrightarrow{T \rightarrow \infty} \frac{ \alpha  x^*_{e}}{2e} 
\end{align*}
}
Denote the optimal solution for the operator's profit LP by $OPT_O$. Let $ALG_O$ be operator's profit obtained by our online algorithm. Using the linearity of expectation we get: $ {ALG_{O}}=\mathbb{E}\Big[ \sum_{e \in E} {\wto n^{x}_{e}p_{e}}\Big]  \geq   \sum_{ e \in E} {\wto  p_{e}}  \frac{ \alpha  x^*_{e}  }{2e} \geq   \sum_{ e \in E} {\wto  p_{e}}\Big( \frac{1}{e} \Big)  \frac{ \alpha  x^*_{e}  }{2} \geq \frac{\alpha}{2e} (OPT_{O})$. Similarly, we can obtain $\frac{\beta}{2e}$ and $\frac{\gamma}{2e}$ competitive ratios for the expected max-min group fairness guarantees on the offline and online sides, respectively.
\end{proof}

\subsection{Proofs for Section \ref{sec:kad}}

\validkad*
\begin{proof}
We shall consider only the operator's profit objective as the other objectives follow through an identical argument. Let $1_{v,t}$ be the indicator random variable for the arrival for vertex type $v$ in round $t$. Then we can obtain a realization and solve the corresponding LP and then take the expected value of LP as an upper bound on the operator's profit objective, i.e. the value $\E_{\mathcal{I}}[LP(G)]$ where $\E_{\mathcal{I}}$ is an expectation with respect to the randomness of the problem. This means replacing $1_{v,t}$ by its realization in the LP below: 
{
\begin{align}
    & \textstyle\max \sum\limits_{t \in [T]} \sum\limits_{e \in E}{\wto \xet }  \label{sys_obj_kad_real} 
    \end{align}
    \vspace{-0.5cm}
    \begin{subequations}\label{constraintkad}
    \begin{align}
     &  \text{s.t} \quad  \forall e \in E, \forall t \in [T]:  0 \leq \xet \leq 1 \label{constraintkad_real : probkad_real } \\
    &  \textstyle\forall u \in U: \sum\limits_{t \in [T]} \sum\limits_{e \in E_u} \xet \leq 1  \label{constraintkad_real: lmatchkad_real} \\
    &  \textstyle\forall v \in V, \forall t \in [T]: \sum_{e \in E_v} \xet \leq 1_{v,t}  \label{constraintkad_real : rmatchkad_real} 
    \end{align}
   \end{subequations}
}
If we were to replace the random variables $1_{v,t}$ by their expected value, then we would retrieve LP(\ref{sys_obj_kad}) where $\E_{\mathcal{I}}[1_{v,t}]=\pvt$. It suffices to show that the value of LP(\ref{sys_obj_kad}) which is the LP value over the ``expected'' graph (the parameters replaced by their expected value) which we now denote by $LP(\E_{\mathcal{I}}[G])$ is an upper bound to $\E_{\mathcal{I}}[LP(G)]$, i.e. $LP(\E_{\mathcal{I}}[G]) \ge \E_{\mathcal{I}}[LP(G)]$. Let $x^{*,G}_{e,t}$ be the optimal solution for a given realization $G$ and $1^G_{v,t}$ be the realization of the random variables over the instance, then we have that $\sum_{e \in E_v} x^{*,G}_{e,t} \leq 1^G_{v,t}$. It follows that $\E_{\mathcal{I}}[x^{*,G}_{e,t}]$ is a feasible solution for LP(\ref{sys_obj_kad}), since $\E_{\mathcal{I}}[\sum_{e \in E_v}  x^{*,G}_{e,t}] \leq \E_{\mathcal{I}}[1^G_{v,t}] = \pvt$ and the rest of the constraints are satisfied as well since they are the same in every realization. However, we have that $\E_{\mathcal{I}}[LP(G)]=\E_{\mathcal{I}}[\sum\limits_{t \in [T]} \sum\limits_{e \in E}{\wto x^{*,G}_{e,t} }] = \sum\limits_{t \in [T]} \sum\limits_{e \in E}{\wto \E_{\mathcal{I}}[x^{*,G}_{e,t}] } \leq \sum\limits_{t \in [T]} \sum\limits_{e \in E}{\wto x^{*}_{e,t} } = LP(\E_{\mathcal{I}}[G])$ where $x^{*}_{e,t}$ is the optimal solution for LP(\ref{sys_obj_kad}) over the ``expected'' graph. The inequality followed since a feasible solution to a problem cannot exceed its optimal solution. 
\end{proof}

\kadalgvalid*
\begin{proof}
We prove the validity of the algorithm for $\lambda=\frac{1}{2}$ by induction. For the base case, it is clear that $\forall e \in E, \ret=1$, hence $\ret 
\ge \lambda =\frac{1}{2}$. Assume for $t'<t$, that $\rho_{e,t'} \ge \lambda=\frac{1}{2}$, then at round $t$ we have: 
\begin{align*}
    1-\ret & = \Pr[\text{$e$ is not available at $t$}]  \\
    & = \Pr[\text{$e$ is matched in $[T-1]$}] \\
    & \leq \sum_{t'<t} \Pr[\text{$e$ is matched in $t'$}] \\
    & = \sum_{t'<t} \Pr[(\text{$e$ is chosen by the algorithm}) \\ 
    & \land (\text{$u$ is unmatched at the beginning of $t$}) \\
    & \land (\text{$v$ arrives at $t$}) ] \\ 
    & = \sum_{t'<t} \pvt \ret (\alpha \frac{\xets}{\pvt}\frac{\lambda}{\ret} + \beta \frac{\yets}{\pvt}\frac{\lambda}{\ret} + \gamma \frac{\zets}{\pvt}\frac{\lambda}{\ret}) \\
    & = \sum_{t'<t} \lambda ( \alpha x^*_{e,t'} + \beta y^*_{e,t'} + \gamma z^*_{e,t'} ) \\
    & \leq \lambda \sum_{t'<t}  ( \alpha x^*_{e,t'} + \beta y^*_{e,t'} + \gamma z^*_{e,t'} ) \\ 
    &  \leq \lambda (\alpha+\beta+\gamma) \leq \lambda \leq \frac{1}{2}
\end{align*}
where we used the fact that $x^*_{e,t'},y^*_{e,t'},z^*_{e,t'} \leq 1$ from constraint (\ref{constraintkad : probkad}) and the fact that $\alpha+\beta+\gamma \leq 1$. From the above, it follows that $\ret \ge \frac{1}{2} \ge \lambda$. 
\end{proof}

Now we restate and prove Theorem \ref{main_th_kad} using Lemma \ref{kad_alg_valid}: 
\mainthkad*
\begin{proof}
For an edge $e$ the probability that it is matched (successfully probed) is the following:
\begin{align*}
    & \Pr[\text{$e$ is successfully probed in round $t$}]  \\ 
    & = \Pr[(\text{$e$ is chosen by the algorithm}) \\ 
    & \land (\text{$u$ is unmatched at the beginning of $t$}) \land (\text{$v$ arrives at $t$}) ] \\ & = \pvt \ret (\alpha \frac{\xets}{\pvt}\frac{\lambda}{\ret} + \beta \frac{\yets}{\pvt}\frac{\lambda}{\ret} + \gamma \frac{\zets}{\pvt}\frac{\lambda}{\ret}) =\\
    & = \alpha \lambda \xets + \beta \lambda \yets + \gamma \lambda \zets \\ 
\end{align*}
Setting $\lambda=\frac{1}{2}$, it follows from the above that $e$ is successfully matched with probability at least $ \frac{1}{2}\alpha \xets$, at least $ \frac{1}{2}\beta \yets$, and at least $ \frac{1}{2} \gamma \zets$. Hence, the guarantees on the competitive ratios follow by linearity of the expectation. 
\end{proof}

\subsection{Proofs for Section \ref{sec:indiv}}
We restate Lemma \ref{reduc_indiv_group} and give its proof: 
\reducgroup* 
\begin{proof}
Given an instance with individual fairness, define $\grps=\{g_1,\dots,g_T\} \cup \{g'_1,\dots,g'_{|U|}\}$ as the set of all groups, thus $|\grps|=T+|U|$, i.e. one group for each time round and one group for each offline vertex. Further given the online side types $V$, create a new online side $V'$ where $|V'|=T|V|$ and $V'=V'_1 \cup V'_2 \dots \cup V'_t\dots \cup V'_T$ where $V'_t$ consists of the same types as $V$. Moreover, $\forall v' \in V'_t, p_{v',t}=\pvt$ and $p_{v',\bar{t}}=0, \forall \bar{t} \in [T]-\{t\}$, finally $\forall v' \in V'_t, g(v')=g_t$. For the offline side $U$, we let each vertex have its own distinct group membership, i.e. for vertex $u_i \in U$, $g(u_i)=g'_i$. 

Based on the above, it is not difficult to see that both problems have the same operator profit, and that the individual max-min fairness objectives of the original instance equal the group max-min fairness objectives of the new instance. 
\end{proof}

From the above Lemma, applying algorithm $\tsfkad$ to the reduced instance leads to the following corollary:
\begin{corollary} \label{corollary_indiv}
Given an instance of two-sided individual max-min fairness, applying $\tsfkad(\alpha,\beta,\gamma)$ to the reduction from Theorem \ref{reduc_indiv_group} leads to a competitive ratio of $(\frac{\alpha}{2},\frac{\beta}{2},\frac{\gamma}{2})$ simultaneously over the operator's profit, the individual fairness objective for the offline side, and the individual fairness objective for the online side, where $\alpha,\beta,\gamma>0$ and $\alpha+\beta+\gamma \leq 1$. 
\end{corollary}
The proof of Theorem \ref{main_th_indiv} is immediate from the above corollary. 

\subsection{Proofs for Theorems \ref{Hardness_th} and \ref{Hardness_th_indiv_group}}
We now restate and prove the hardness result of Theorem \ref{Hardness_th}:
\Hardnessth*
\begin{proof}
We prove it for group fairness in the \textbf{KIID} setting, since the \textbf{KIID} setting is a special case of the \textbf{KAD} setting, then this also proves the upper bound for the \textbf{KAD} setting.

Consider the graph $G=(U,V,E)$ which consists of three offline vertices and three online vertex types, i.e. $|U|=|V|=3$. Each vertex in $U$ ($V$) belongs to its own distinct group. The time horizon $T$ is set to an arbitrarily large value. The arrival rate for each $v \in V$ is uniform and independent of time, i.e. \textbf{KIID} with $\pv=\frac{1}{3}$. Further, the bipartite graph is complete, i.e. each vertex of $U$ is connected to all of the vertices of $V$ with $p_e=1$ for all $e \in E$. We also let $\Delta_u =1$ for each $u \in U$, $n_v =\frac{T}{3}$ and $\Delta_v=1$ for each $v \in V$. We represent the utilities on the edges of $E$ with matrices where the $(i,j)$ element gives the utility of the edge connecting vertex $u_i \in U$ and vertex $v_j \in V$. The utility matrices for the platform operator, offline, and online sides are following, respectively: 
$$M_O =  \begin{bmatrix}
 1 & 0 &0 \\
 0 &1 &0 \\
 0 & 0 &1 
\end{bmatrix} ,M_U =  \begin{bmatrix}
 0 & 0 &1 \\
 1 & 0  &0 \\
 0 & 1 & 0 
\end{bmatrix} , M_V=  \begin{bmatrix}
 0 & 1 &0 \\
 0 &0 &1 \\
 1 & 0 &0 
\end{bmatrix}.$$
It can be seen that the utility assignments in the above example conflict between the three entities. 

Let $\OPT_O, \OPT_U,$ and $\OPT_V$ be the optimal values for the operator's profit, offline group fairness, and online group fairness, respectively. It is not difficult to see that $\OPT_O=3$, $\OPT_U=1$, and $\OPT_V=1$.
Now, denote by $A,B,$ and $C$ the edges with values of 1 for $M_O,M_U,$ and $M_V$ in the graph, respectively. Further, for a given online algorithm, let $a_j,b_k,$ and $c_{\ell}$ be the expected number of probes received by edges $j \in A, k \in B,$ and $\ell \in C$, respectively. Moreover, denote the algorithm's expected value over the operator's profit, expected fairness for offline vertices, and expected fairness for online vertices by $\ALG_O,\ALG_U$, and $\ALG_V$, respectively. We can upper bound the sum of the competitive ratios as follows:

\begin{align*}
    & \frac{\ALG_O}{\OPT_O} + \frac{\ALG_U}{\OPT_U} + \frac{ALG_V}{\OPT_V}
    \\
    & \leq \frac{\sum_{j \in A} a_j}{3} + \frac{\min_{k \in B} b_j}{1} + \frac{\min_{\ell \in C} c_j}{1} \\ 
    & \leq \frac{\sum_{j  \in A} a_i}{3} + \frac{\big(\sum_{k \in B} b_i \big)/3}{1} + \frac{\big(\sum_{\ell \in C} c_i \big)/3}{1} \\ 
    & \leq \frac{\sum_{j  \in A} a_i + \sum_{k  \in B} b_i  + \sum_{\ell  \in C} c_i }{3}    \leq \frac{3}{3} = 1 
\end{align*}

in the above, the second inequality follows since the minimum value is upper bounded by the average. The last inequality follows since $\Delta_u=1$ and therefore the expected number of probes any offline vertex receives cannot exceed 1 and we have $|U|=3$ many vertices. 

To prove the same result for individual fairness we use the same graph. We note that the arrival of vertices in $V$ is \textbf{KAD} instead with the $i^{\text{th}}$ vertex $v_i$ having $p_{v_i,i}=1$ and $p_{v_i,t}=0, \forall t\neq i$. Then we follow an argument similar to the above. 
\end{proof}
The following proves Theorem \ref{Hardness_th_indiv_group} therefore showing that there is indeed a conflict between achieving group and individual fairness even if we were to consider only one side of the graph.

\Hardnessthindivgroup*
\begin{proof}
Let us focus on the offline side, i.e. we consider $\ag$ and $\ai$ that are the competitive ratios for the group and individual fairness of the offline side. 

Consider a graph which consists of two offline vertices and one online vertex, i.e. $|U|=2$ and $|V|=1$. Further, there is only one group. Let $p_e=1, \forall e\in E$ and $\forall u\in U, \forall v \in V: \pl=\pr=1$. $U$ has two vertices $u_1$ and $u_2$ both connected to the same vertex $v \in V$. For edge $(u_1,v)$, we let $w^U_{(u_1,v)}=1$ and for edge $(u_2,v)$, we let $w^U_{(u_2,v)}=L$ where $L$ is an arbitrarily large number. Note that both of these weights are for the utility of the offline side. Finally, we only have one round so $T=1$. 

Let $\theta_1$ and $\theta_2$ be the expected number of probes edges $(u_1,v)$ and $(u_2,v)$ receive, respectively. Note that $\theta_1=1-\theta_2$. It follows that the optimal offline group fairness objective is $\OPT^U_{G}=\max\limits_{\theta_1, \theta_2} ( \theta_1+L\theta_2) = \max\limits_{\theta_2} ( (1-\theta_2)+L\theta_2) = L$. Further, the optimal offline individual fairness objective is $\OPT^U_{I}= \min\{\theta_1,L\theta_2\}$, it is not difficult to show that $\OPT^U_{I}=\frac{L}{L+1}$. Now consider the sum of competitive ratios, we have: 
\begin{align*}
    \frac{\ALG^U_G}{\OPT^U_{G}} + \frac{\ALG^U_I}{\OPT^U_{I}} & = \frac{\theta_1+L\theta_2}{L} + \frac{\min\{\theta_1,L\theta_2\}}{\frac{L}{L+1}} \\ 
    & \leq \frac{\theta_1+L\theta_2}{L} + \frac{\theta_1(L+1)}{L} \\ 
    & = \frac{(L+2)\theta_1+L\theta_2}{L} \\ 
    & = (\theta_1+\theta_2) + \frac{2\theta_1}{L} \\ 
    & \leq 1 + \frac{2\theta_1}{L} \xrightarrow{L\rightarrow \infty} 1 
\end{align*}
this proves the result for the offline side of the graph. 

To prove the result for the online side, we reverse the graph construction, i.e. having one vertex in $U$ and two vertex types in $V$ which arrive with equal probability. It now holds that $\OPT^V_{I}= \min\{\theta_1,L\theta_2\}$ and by setting $T$ to an arbitrarily large value $\OPT^V_{G}=L$. Then we follow an identical argument to the above. 
\end{proof}

\section{Additonal Experimental Results}
As mentioned before one of the major contributions of our work is that we consider the operator's profit and fairness for both sides \emph{simultaneously} instead of fairness for only one side. To further see the effects of ignoring one side, we run $\tsf$ with one side ignored (see table \ref{table:one_side_ignored}). It is clear that the fairness objective for the ignored side is indeed lower in comparison to what can be achieved in figure \ref{fig:fm_exp_fig}. More precisely, we can see that the Offline (Driver) and Online (Rider) fairness can be simultaneously improved to around $0.5$ by setting $\alpha=0.5, \beta=\gamma=0.25$ (figure \ref{fig:fm_exp_fig}) whereas their values when their optimization weight is set to zero is $0.387$ and $0.41$ , respectively (see table \ref{table:one_side_ignored}).

\begin{table}[h!]

\begin{tabular}{cccc}
\hline
                             & \multicolumn{1}{l}{\textbf{Profit}} & \multicolumn{1}{l}{\textbf{Driver  Fairness}} & \multicolumn{1}{l}{\textbf{Rider Fairness}} \\ \hline
\textbf{$\alpha = \gamma = 0.5, \beta = 0$}           & 0.43            & 0.387                                                              & 0.509                                      \\
\textbf{$\alpha = \beta = 0.5, \gamma = 0$} & 0.564           & 0.498                                                              & 0.41                                                                  
\end{tabular}
\caption{Results of running $\tsf$ on the NYC dataset with the fairness on one side ignored, i.e. its optimization weight set to 0: (\textbf{Top Row}) Offline (Driver) fairness ignored ($\alpha = \gamma=0.5, \beta = 0$) and (\textbf{Bottom Row}) Online (Rider) fairness ignored ($\alpha = 0.5, \beta = 0.5, \gamma = 0$).}
\label{table:one_side_ignored}
\end{table}

\end{document}